\documentclass[11pt]{amsart}
\usepackage{amscd}
\usepackage{amsmath}
\usepackage{amssymb}
\usepackage{mathtools}
\usepackage{latexsym}
\usepackage{lscape}
\usepackage{xcolor}
\usepackage[colorlinks=true, urlcolor=dblue, linktocpage=true, linkcolor=dblue,citecolor=purple]{hyperref}

\colorlet{dblue}{blue!70!black}
\colorlet{dgreen}{green!60!black}

\newtheorem{thm}{Theorem}[section]

\newtheorem{lem}[thm]{Lemma}

\newtheorem{lem-def}[thm]{Lemma-Definition}

\theoremstyle{remark}

\newtheorem{rmk}{Remark}[section]
\theoremstyle{definition}

\newtheorem{defn}{Definition}[section]

\numberwithin{equation}{section}

\newcommand{\quash}[1]{}  
\newcommand{\nc}{\newcommand}
\nc{\on}{\operatorname}

\newcommand{\bC}{{\mathbb C}}

\newcommand{\bK}{{\mathbb K}}

\newcommand{\bQ}{{\mathbb Q}}
\newcommand{\bR}{{\mathbb R}}

\newcommand{\bZ}{{\mathbb Z}}

\newcommand{\DD}{{\mathcal D}}

\newcommand{\calF}{{\mathcal F}}

\newcommand{\FF}{{\mathcal F}}

\nc{\al}{{\alpha}} 
\nc{\ga}{{\gamma}}
\nc{\ve}{{\varepsilon}} \nc{\Ga}{{\Gamma}} \nc{\la}{{\lambda}}
\nc{\La}{{\Lambda}}

\nc{\ad}{{\on{ad}}}

\nc{\aff}{{\on{aff}}}
\nc{\Aff}{{\mathbf{Aff}}}

\nc{\Bun}{{\on{Bun}}}

\nc{\der}{{\on{der}}}

\nc{\diag}{{\on{diag}}}

\nc{\Fl}{{\calF\ell}}

\nc{\Hol}{{\on{Hol}}}

\nc{\Id}{{\on{Id}}}

\nc{\Ind}{{\on{Ind}}}

\nc{\res}{{\on{res}}}

\newcommand{\sgn}{{\on{sgn}}}

\nc{\tr}{{\on{tr}}}
\newcommand{\Tr}{{\on{Tr}}}

\nc{\GSp}{{\on{GSp}}} \nc{\GU}{{\on{GU}}} \nc{\SL}{{\on{SL}}}
\nc{\SU}{{\on{SU}}} \nc{\SO}{{\on{SO}}}

\nc{\four}{{\calF our}}

\newcommand\bbe{\begin{equation}}
\newcommand\be{\begin{equation}}
\newcommand\ba{\begin{eqnarray}}
\newcommand\ee{\end{equation}}
\newcommand\ea{\end{eqnarray}}

\newcommand{\arxiv}[1]
  {\href{http://arxiv.org/abs/#1}{arXiv:#1}}

\newcommand{\lb}{\left(}
\newcommand{\rb}{\right)}
\newcommand{\lsb}{\left[}
\newcommand{\rsb}{\right]}

\newcommand{\nn}{\nonumber}
\newcommand{\mrm}{\mathrm}

\def\question#1{{}}

\title{Quadratic reciprocity from a family of adelic conformal field theories }
\author{An Huang}
\address{An Huang, Department of Mathematics, Brandeis University, Waltham, MA 02453, USA}
\author{Bogdan Stoica}
\address{Bogdan Stoica, Department of Physics \& Astronomy, Northwestern University, Evanston, IL 60208, USA}
\author{Xiao Zhong}
\address{Xiao Zhong, Department of Pure Mathematics, University of Waterloo,  Waterloo, Ontario, N2L 3G1, Canada}

\calclayout

\begin{document}

{\noindent \small nuhep-th/21-10}\vspace{0.85cm}

\maketitle

\begin{abstract}
We consider a deformation of the two-dimensional free scalar field theory by raising the Laplacian to a positive real power. It turns out that the resulting non-local generalized free action is invariant under two commuting actions of the global conformal symmetry algebra, although it is no longer invariant under the full Witt algebra. Furthermore, there is an adelic version of this family of conformal field theories, parameterized by the choice of a number field, together with a Hecke character. Tate's thesis gives the Green's functions of these theories, and ensures that these Green's functions satisfy an adelic product formula. In particular, the local $L$-factors contribute to the prefactors of these Green's functions. Quadratic reciprocity turns out to be a consequence of an adelic version of a holomorphic factorization property of this family of theories on a quadratic extension of $\mathbb{Q}$. We explain that at the Archimedean place, the desired holomorphic factorization follows from the global conformal symmetry. 
\end{abstract}

\setcounter{tocdepth}{2}
\tableofcontents

\section{Introduction}

\noindent The main purpose of this paper is to provide a physics framework to understand the classical quadratic reciprocity law in number theory, in a way suited for generalizations. For this purpose, we will aim for a physics interpretation of the analytic statement of quadratic reciprocity, using Dirichlet $L$-functions and Dedekind zeta functions. 

It is well-known that the $p$-adic open string worldsheet theory in genus zero has an asymptotic boundary dual theory, which is a non-local generalized free field theory on $\mathbb{Q}_p$ with global conformal symmetry. This is the $p$-adic counterpart of the fact that in Archimedean string theory, after one locally fixes the worldsheet metric to be flat, the Polyakov action enjoys a remnant gauge symmetry given by the global conformal group. 

This non-Archimedean field theory has an action given by Eq. \eqref{eq22} below, where $D_1$ is a so-called Vladimirov derivative operator. It is known that $D_1$ can be written as the Fourier conjugate of a continuous multiplicative quasi-character of $\mathbb{Q}_p^\times$. From this, in \cite{Huang:2020aao}, a generalization of this non-Archimedean field theory was considered, where one replaces the particular Vladimirov derivative $D_1$, by the Fourier conjugate of a general quasi-character. This setup immediately generalizes to any characteristic zero local field, where the theory is determined by the choice of the quasi-character. Next, one observes that such theories are still generalized free theories, meaning that Wick's theorem holds, i.e. the action consists of only a kinetic term, given by a self-adjoint bilinear operator applied to a scalar field. In addition, it turns out that the theories still respect the global conformal group ({\bf Theorem} \ref{invariant}), and the two-point function is given by the local functional equation in the classical Tate's thesis. We will explain these points in detail in Section \ref{section2}.

As a next step, given a number field and a Hecke character, for each place, one obtains such a field theory specified by the local component of the given Hecke character. Another observation in \cite{Huang:2020aao} is that the global functional equation in Tate's thesis is equivalent to the adelic product formula of the Green's functions at each place. We call this an adelic field theory, where the field theories at each place are adelically compatible, in the sense that the adelic product formula of the two-point functions holds.

One also observes that the local $L$-factors enter as prefactors of the two-point functions, which makes it possible to formulate certain $L$-function identities in terms of the two-point functions. This is the bridge we shall use to translate the quadratic reciprocity law into a statement of physics, regarding these two-point functions. For this purpose, we will need to consider a family of adelic field theories on a quadratic extension of $\mathbb{Q}$, parameterized by an unramified Hecke character.

The physics translation of quadratic reciprocity in this context is given by the {\bf main observation} \eqref{main} in Section \ref{secquadreciproc}. After explaining the main observation and why it is equivalent to quadratic reciprocity, and calculations of certain "epsilon-factors" that come into play, we shall explain that it follows from a proposed adelic version of a holomorphic factorization property of our family of adelic field theories. This property states that there exists a holomorphic factorization of the scalar field $\phi$ at each place, such that the factorizations are adelically compatible, in a sense that shall be explained in Section \ref{factorization}. Therefore, quadratic reciprocity follows from this adelic factorization property for our family of theories. Furthermore, we physically derive the factorization property at the Archimedean place: at the complex place, it turns out that the resulting theories are non-local deformations of two-dimensional free scalar Euclidean field theory, by raising the Laplacian to a positive real power $0<s<1$. We show that for such deformations, although the action is no longer invariant under the full Witt algebra of local conformal transformations, it still invariant under two commuting actions (holomorphic and anti-holomorphic) of the global conformal algebra ({\bf Theorem} \ref{CFT}). These two copies of the global conformal algebra turn out to be powerful in controlling the quasi-primary fields, as a consequence of the infinite dimensional representations involved.  As a consequence, the global conformal symmetry ensure the desired holomorphic factorization. As a side remark, we also explain the hidden role of this deformation in dimensional regularization.

We expect the above to generalize to non-Archimedean places, thus providing a physics explanation of the quadratic reciprocity law in this context. We leave it to a future paper to fill out the remaining details regarding the holomorphic factorization at non-Archimedean places. 

There exist other physical interpretations of quadratic reciprocity, using e.g. adelic quantum mechanics with a related but simpler idea \cite{vvzbook}. Furthermore, Tate's thesis has also found a role in the context of $p$-adic bosonic string tachyon scattering amplitudes \cite{adelicNpoint}, although a complete conceptual understanding of this role is still lacking. In the present paper we will exhibit relations between number theory and physics through our family of adelic conformal field theories, where we start to see a variety of ingredients of number theory and physics coming together in a systematic and intricate way. As a consequence, our framework is aimed for  generalizations and future developments.

\textbf{Acknowledgments:} We thank Jo\"el Bella\"iche, Jonathan Heckman, John Joseph Carrasco, Matthew Headrick, Bong Lian, Omer Offen, and Shing-Tung Yau for useful discussions. The work of A.H. was supported by the Simons Collaboration for Mathematicians under Award Number 708790. The work of B.S. was supported by the Department of Energy under Award Number DE-SC0021485. Part of the work was done in the Harvard Center of Mathematical Sciences and Applications.

\section{Worldsheet boundary theory}
\label{section2}

\subsection{Vladimirov derivative theories}

\noindent The $p$-adic open string worldsheet theory in flat target has a dual theory given by the action (up to an overall constant) \cite{zabrodin,zabrodin2,Gubser:2016guj,Heydeman:2016ldy}

\begin{equation}\label{eq21}
S=\frac{\Gamma_p(\pi_2)}{4}\int_{\mathbb{Q}_p}\int_{\mathbb{Q}_p} \frac{(\phi(z)-\phi(x))^2}{|z-x|_p^2}dzdx,
\end{equation}
where $\phi$ is a real scalar field, and $dz,dx$ is a choice of the local Haar measure on $\mathbb{Q}_p$. The constant $\Gamma_p(\pi_2)$ shall be explained below in a moment. 

The theory is conformal, in the sense that $S$ is invariant under the action of the global conformal group $G=\mathrm{GL}(2,\mathbb{Q}_p)$ on $\phi$ via fractional linear transformations of the argument. This is of basic importance, e.g. in calculating string scattering amplitudes. In Section \ref{section22} we will show that this fact holds in general, for a generalized free field theory defined on a characteristic zero local field parameterized by a choice of a quasi-character as first considered in \cite{Huang:2020aao}. Below we summarize the definition of these theories.
 
Equation \eqref{eq21} can be rewritten as
\be
\label{eq22}
S=\frac{1}{2}\int_{\mathbb{Q}_p} \phi(x) D_1 \phi(x)dx,
\ee 
where $D_1$ is a {\bf regularized Vladimirov derivative} acting on functions in the Bruhat-Schwartz space of $\mathbb{Q}_p$, i.e. on compactly supported locally constant functions on $\mathbb{Q}_p$.

Assuming $\Re(s)> 0$, we define more generally 
\be
D_{s}\phi(z):=\Gamma_p(\pi_{s}\pi_1)\int_{\bQ_p}\frac{\phi(x)-\phi(z)}{|x-z|_p^{s+1}}dx,
\ee
where $\Gamma_p$ is the $p$-adic Gelfand-Graev Gamma function specified by a choice of an additive character of $\mathbb{Q}_p$, $\pi_{s}: k\to |k|_p^{s}$ is a multiplicative character on $\bQ_p^\times$, and $s=1$ was used in the above action \eqref{eq22}. It is known that $D_s=\mathcal F\pi_s\mathcal F^{-1}$, where $\mathcal F$ is the Fourier transform w.r.t. the same choice of the additive character $\psi$ of $\mathbb{Q}_p$ defining $\Gamma_p$.

Fermionic theories were studied by \cite{Gubser:2017qed}, which considered the action
\be
\label{actionGubser}
S= \frac{1}{2} \int_{\mathbb{Q}_p} \phi(x) D_{1,\tau} \phi(x),
\ee
where
\be
D_{s,\tau}f(z):=\Gamma_p(\pi_{s,\tau}\pi_1)\int_{\bQ_p}\frac{\phi(x)-\phi(z)}{|x-z|_p^{s+1}(x-z,\tau)_p}dx.
\ee
Here $\tau\in\mathbb{Q}_p^\times$, and $(x,\tau)_p$ is the Hilbert symbol: $(x,\tau)_p=1$ iff $c^2=xa^2+\tau b^2$ has a nonzero solution over $\mathbb{Q}_p$, and equals $-1$ otherwise. Furthermore, $\pi_{s,\tau}(x)\coloneqq \pi_s(x) \lb x,\tau \rb_p$. A key feature in the construction of \cite{Gubser:2017qed} was that one needed a $p$-adic version of the derivative operator $\partial/\partial t$, which is the Fourier conjugate of the absolute value function times the sign function.

The construction of physical theories in terms of Vladimirov derivatives, as in \cite{Gubser:2017qed}, can be immediately extended to a more general number field $\mathbb{K}$. Let $\nu$ be a non-Archimedean place of~$\mathbb{K}$, and let $\chi_{\nu}$ be a quasi-character of $\mathbb{K}_{\nu}^\times$. Then $\chi_{\nu}$ can be written as $\chi_{\nu}={(\chi_{s})}_\nu\tilde{\chi}_{\nu}$, where ${(\chi_{s})}_\nu(x)=|x|_{\nu}^s$, and $\tilde{\chi}_{\nu}$ is the finite unitary part of the character. For a complex parameter $s\in\mathbb{C}$ with $\Re(s)> 0$, the generalized Vladimirov derivative associated to $\chi_{\nu}$ acting on the Bruhat-Schwartz space $S(\mathbb{K}_{\nu})$ is defined as \cite{Huang:2020aao}
\be
\label{Eqq2666}
D_{\chi_{\nu}}\phi(z):=\Gamma_{\mathbb{K}_{\nu}}\lb\chi_{\nu}(\chi_{1})_\nu\rb\int_{\mathbb{K}_{\nu}}\frac{\phi(x)-\phi(z)}{\chi_{\nu}(x-z)|x-z|_{\nu}}dx.
\ee

Just as in the $\mathbb{Q}_p$ case, it turns out that there is an alternative equivalent definition of the Vladimirov derivative via the Fourier conjugate of character $\chi_\nu$:

\begin{lem}[Huang-Stoica-Yau-Zhong '19]
\label{lemma2p1}
Assume $\Re(s)>0$, then $D_{\chi_{\nu}}=\mathcal F\chi_{\nu}\mathcal F^{-1}.$
\end{lem}
The proof of Lemma \eqref{lemma2p1} follows from elementary computations with ball characteristic functions, involving the interplay of non-Archimedean norms, group characters and continuity of character $\chi_{\nu}$ (see \cite{Huang:2020aao}).

\subsection{Vladimirov derivative theories for characteristic zero local fields}
\label{section22}

For a quasi-character $\chi_\nu$, the action \eqref{actionGubser} generalizes as
\be
\label{hereisS}
S = \frac{1}{2}\int_{\mathbb{K}_\nu} \phi(x) D_{\chi_\nu} \phi(x) dx.
\ee
Here $\phi$ is a real scalar field on $\mathbb{K}_{\nu}$, and $D_{\chi_\nu}$ is the Vladimirov derivative given by Eq.~\eqref{Eqq2666}.

\begin{rmk}
For a general character $\chi_\nu$, the action in Eq. \eqref{hereisS} will in general be complex-valued. However, if $s>0$ and the finite part of the character is quadratic (i.e. takes values in $\{\pm 1\}$), then the action \eqref{hereisS} is real. Later we will consider a field $\phi$ that is complex-valued; see Lemma \ref{lemma2point2} below for the reality condition in that case.
\end{rmk}

\begin{rmk}
When $\tilde\chi_\nu(-1) = -1$ and the field $\phi$ is bosonic, the action \eqref{hereisS} vanishes. This is because in this case the Vladimirov derivative is anti-self-adjoint, so a change of variables $x\to-x$ in the integration over $\mathbb{K}_\nu$ in Eq. \eqref{hereisS} sends $S\to -S$. To construct a nonvanishing action for the case $\tilde\chi_\nu(-1) = -1$, one can either consider field $\phi$ to be fermionic, so that it anti-commutes (as in \cite{Gubser:2017qed}), or promote $\phi$ to a complex-valued (i.e. two-component) scalar field. In the present paper we will mostly consider the second option, that of a complex scalar field.
\end{rmk}
 
Next, we consider the physics action $S$ on $\mathbb{K}_\nu$ for a complex scalar field $\phi$,
\be\label{complex action}
 S = \int_{\mathbb{K}_\nu} \Bar{\phi} D_{\chi_\nu} \phi.
\ee
By Lemma \ref{lemma2p1}, an explicit expression of the physics action is given by
\bbe
\label{eq18}
S = \Gamma\lb\chi_\nu {(\chi_1)}_\nu \rb\int_{\mathbb{K}_\nu} \int_{\mathbb{K}_\nu} \frac{\Bar{\phi}(x)(\phi(x') - \phi(x))}{\tilde{\chi}_\nu(x'-x)|x'-x|^{s+1}} dx' dx.
\ee

We have the following basic lemmas.

\begin{lem}
\label{lemma2point2}
For positive real $s$ and quadratic unitary character $\tilde{\chi}_\nu$, the Vladimirov derivative $D_{\chi_\nu}$ is self-adjoint w.r.t. the Hermitian integral pairing, and the physics action is real, recalling that $\chi_{\nu}={(\chi_{s})}_\nu\tilde{\chi}_{\nu}$.
\end{lem}
\begin{proof}
We check this by direct computation with the double Fourier conjugate formulation of the derivative operators,
\ba
 \langle \phi_1, D_{\chi_\nu}\phi_2 \rangle &=& \int_{\mathbb{K}_\nu} \overline{\phi}_1(x) D_{\chi_\nu} \phi_2(x) dx \\
 \label{eq21111}
 &=& \int_{\mathbb{K}_\nu} \overline{\phi}_1(x) \int_{\mathbb{K}_\nu} e^{2\chi i k x} (\chi_s)_\nu(k) \int_{\mathbb{K}_\nu} e^{-2\pi i kx'} \phi_2(x') dx' dk dx.
\ea
Notice that since we are working with quadratic character $\tilde{\chi}_\nu$ and real $s$, we have $\overline{\tilde{\chi}_\nu(x)|x|^s} = \tilde{\chi}_\nu(x)|x|^s = (\chi_s)_\nu(x)$, so that
\be
\overline{D\phi}_1(x) = \int_{\mathbb{K}_\nu} e^{-2\pi i k x} (\chi_s)_\nu(k) \int_{\mathbb{K}_\nu} e^{2\pi i kx'} \overline\phi_1(x') dx' dk dx,
\ee
so that Eq. \eqref{eq21111} implies
\ba
\langle \phi_1, D_{\chi_\nu}\phi_2 \rangle =\langle D_{\chi_\nu} \phi_1, \phi_2 \rangle.
\ea 
Then
\bbe
S = \langle \phi, D_{\chi_\nu}\phi \rangle = \langle D_{\chi_\nu} \phi, \phi \rangle = \bar{S},
\ee 
therefore the physics action $S$ is real.
\end{proof}

In terms of the Euclidean integral pairing, we have the following.
\begin{lem}
Under the Euclidean integral pairing, $D_{\chi_\nu}$ above is self-adjoint if $\tilde{\chi}_\nu(-1) =- 1$, and anti-self-adjoint otherwise.
\end{lem}
\begin{proof}

By a similar computation,
\ba
\langle \phi_1, D_{\chi_\nu}\phi_2 \rangle_{E} &=& \int_{\mathbb{K}_\nu} \phi_1(x) D_{\chi_\nu} \phi_2(x) dx \\
&=&  \int_{\mathbb{K}_\nu} \phi_1(x) \int_{\mathbb{K}_\nu} e^{2\pi i k x} (\chi_s)_\nu(k) \int_{\mathbb{K}_\nu} e^{-2\pi i kx'} \phi_2(x') dx' dk dx \\
&=& (\chi_s)_\nu(-1) \int_{\mathbb{K}_\nu} \phi_1(x) \int_{\mathbb{K}_\nu} e^{-2\pi i (-k) x} (\chi_s)_\nu(-k) \times \\
& &\times\int_{\mathbb{K}_\nu} e^{2\pi i (-k)x'} \phi_2(x') dx' dk dx \nn \\
&=& (\chi_s)_\nu(-1) \langle D_{\chi_\nu}\phi_1, \phi_2 \rangle_E.
\ea
\end{proof}

Next, we show that the physics action \eqref{eq18} enjoys global conformal symmetry. For this, we consider a central extension of $\mathrm{GL}(2,\mathbb{K}_{\nu})$ denoted by $\tilde{G}$. The group $\tilde{G}$ as a set is equal to the set of tuples $(g, \epsilon)$ where $g \in \mrm{GL}(2,\mathbb{K}_{\nu})$, and $\epsilon = \tilde{\chi}_\nu^{1/2}(\det(g))$ is a chosen branch. The group operation is 
 \be 
 (g_1 , \epsilon_1) (g_2, \epsilon_2) = (g_1 g_2 , \epsilon)
 \ee 
 where $\epsilon(g_1g_2) = \epsilon_1(g_1)\epsilon_2(g_2)$. The associative condition can be checked directly, that is 
 \be
 \epsilon_{12}(g_1g_2)\epsilon_3(g_3) = \epsilon_1(g_1)\epsilon_2(g_2)\epsilon_3(g_3) = \epsilon_1(g_1)\epsilon_{23}(g_2g_3),
 \ee 
 where $g_1,g_2,g_3 \in \mrm{GL}(2,\mathbb{K}_{\nu})$, $\epsilon_{12}(g_1g_2) = \epsilon_1(g_1)\epsilon_2(g_2)$ and $\epsilon_{23}(g_2g_3) = \epsilon_2(g_2)\epsilon_3(g_3)$.
 
 We consider the following action of $\tilde{G}$ on $\phi(x)$,
 \be\label{action}
 \phi(x)\cdot (g, \epsilon) = \phi(g(x)) \chi(cx + d)|cx + d|^{-1} |det(g)|^{(1-s)/2}\epsilon^{-1}(g)
 \ee 
where $g = \begin{bmatrix}
a & b \\
c & d 
\end{bmatrix}$. 

\begin{rmk}
In terms of physics, \eqref{action} is a generalization of an Archimedean quasi-primary field to an arbitrary place, which shall play an important role in Section \ref{factorization}.
\end{rmk}
\begin{rmk}
Note that the group action doesn't preserve the Schwartz space. Because of this, one needs to enlarge the field space to be e.g. the minimum one spanned by Schwartz space of $\mathbb{K}_\nu$ acted on by~$\tilde{G}$.
\end{rmk}

We now state the theorem.

\begin{thm} \label{invariant}
The physics action \eqref{eq18}, or equivalently \eqref{complex action}, is invariant under the action \eqref{action} of $\tilde{G}$ on $\phi(x)$, where $\chi(x) = \tilde{\chi}_\nu(x) |x|^s$, $0<s<1$ and $\tilde{\chi}_\nu(x)$ is the unitary part of the multiplicative character.
\end{thm}

\begin{proof}
By considering real and imaginary parts of $\phi(x)$, it suffices to prove that the integral
\bbe 
\tilde{S} = \int_{\mathbb{K}_\nu} \int_{\mathbb{K}_\nu} \frac{\phi_1(x)(\phi_2(x') - \phi_2(x))}{\tilde{\chi}(x'-x)|x'-x|^{s+1}} dx' dx 
\ee 
is invariant under the action of $\tilde{G}$, where $\phi_1$, $\phi_2$ are two arbitrary real-valued compactly supported locally constant functions. Notice that $\tilde{S}$ is clearly invariant under the action of $Z(\tilde{G})$ on $\phi(x)$. We consider the action of generators of $\tilde{G}$. We start with
\be
g =\begin{bmatrix}
1 & b \\
0 & 1 
\end{bmatrix}.
\ee
We have
\ba
\tilde{S}\cdot(g, \epsilon) &=& \int_{\mathbb{K}_\nu} \int_{\mathbb{K}_\nu} \frac{\phi_1(x)\cdot (g,\epsilon)(\phi_2(x')\cdot(g,\epsilon) - \phi_2(x)\cdot(g,\epsilon))}{\tilde{\chi}(x'-x)|x'-x|^{s+1}} dx' dx \\
&=& \int_{\mathbb{K}_\nu} \int_{\mathbb{K}_\nu} \frac{\phi_1(x+b)(\phi_2(x'+b) - \phi_2(x+b))}{\tilde{\chi}(x'-x)|x'-x|^{s+1}} dx' dx \\
&=& \tilde{S}.
\ea

Next we consider \be g =\begin{bmatrix}
0 & 1 \\
1 & 0 
\end{bmatrix}.
\ee 
We have
\bbe
\tilde{S}\cdot(g,\epsilon) = \int_{\mathbb{K}_\nu} \int_{\mathbb{K}_\nu} \frac{\phi_1\lb\frac{1}{x}\rb \chi_\nu(x)|x|^{-1}\epsilon^{-1} \lb\phi_2\lb\frac{1}{x'}\rb\chi_\nu(x')|x'|^{-1} \epsilon^{-1} - \phi_2\lb\frac{1}{x}\rb\chi_\nu(x)|x|^{-1} \epsilon^{-1}\rb}{\chi_\nu(x'-x)|x'-x|} dx' dx.
\ee 
Let $z \coloneqq 1/x$, $z' \coloneqq 1/x'$, then
\bbe
\tilde{S}\cdot(g,\epsilon) = \tilde{\chi}_\nu(-1)\int_{\mathbb{K}_\nu} \int_{\mathbb{K}_\nu} \frac{\phi_1(z) \chi_\nu\lb\frac{1}{z}\rb|z|\lb\phi_2(z')\chi_\nu\lb\frac{1}{z'}\rb|z'| - \phi_2(z)\chi_\nu(\frac{1}{z})|z|\rb}{\chi_\nu\lb\frac{z- z'}{zz'}\rb\left|\frac{z - z'}{zz'}\right|} |zz'|^{-2}dz' dz.
\ee
The difference 
\bbe
\tilde{S} - \tilde{S}\cdot(g,\epsilon) = - \tilde{\chi}_\nu(-1)\int_{\mathbb{K}_\nu} \int_{\mathbb{K}_\nu} \frac{\phi_1(z)\phi_2(z)}{\chi_\nu(z-z')|z -z'|}\lb 1 - \frac{\chi_\nu\lb\frac{z'}{z}\rb}{\left|\frac{z'}{z}\right|}\rb dz' dz.
\ee 
By using our integral formula for the Vladimirov derivative \eqref{lemma2p1}, the difference $\tilde{S}-\tilde{S}\cdot(g,\epsilon)$ can be written as
\be
\tilde{S}-\tilde{S}\cdot(g,\epsilon)=\int_{\mathbb{K}_\nu} \phi_1(z)\phi_2(z)|z|^{1-s}\tilde{\chi}_\nu(z)^{-1}\lb \frac{D_{(\chi_s)_\nu\tilde{\chi}_\nu}}{\Gamma((\chi_{s+1})_\nu\tilde{\chi}_\nu)}(\chi_{s-1})_\nu\tilde{\chi}_\nu\rb(z)dz.
\ee
By the local functional equation for zeta distributions in Tate's thesis, we have 
\ba
(D_{(\chi_s)_\nu\tilde{\chi}_\nu}(\chi_{s-1})_\nu\tilde{\chi}_\nu)(z) &=& \FF(\chi_s)_\nu\tilde{\chi}_\nu\FF^{-1}((\chi_{s-1})_\nu\tilde{\chi}_\nu) \\
&=& \FF(\chi_s)_\nu\tilde{\chi}_\nu P \FF((\chi_{s-1})_\nu\tilde{\chi}_\nu)\\
&=& \FF(\chi_s)_\nu\tilde{\chi}_\nu P \Gamma((\chi_s)_\nu\tilde{\chi}_\nu)(\chi_{-s})_\nu\tilde{\chi}_\nu^{-1}\\
&=& \tilde{\chi}_\nu^{-1}(-1)\Gamma((\chi_s)_\nu\tilde{\chi}_\nu)\delta(z),
\ea
where $Pf(x):=f(-x)$ is the parity change operator, and recall that $\FF^2=P$. Therefore, when $0<\Re(s)<1$, $\tilde{S}-\tilde{S}\cdot (g,\epsilon)=~0$.

Finally we are left to check the action of
\be
g = \begin{bmatrix}
a & 0 \\
0 & 1 
\end{bmatrix}.
\ee 
A direct computation involving a change of variables shows that $\tilde{S}$ is invariant:
\ba
\tilde{S}\cdot(g,\epsilon) &=& \int_{\mathbb{K}_\nu} \int_{\mathbb{K}_\nu} \frac{\phi_1(ax)(\phi_2(ax') - \phi_2(ax)) \chi_\nu^{-1}(a)|a|}{\chi_\nu(x' - x)|x' -x|} dx' dx \\
 &=& \int_{\mathbb{K}_\nu} \int_{\mathbb{K}_\nu} \frac{ \phi_1(z) (\phi_2(z') - \phi_2(z))}{\chi_\nu(z' - z)|z - z'|} dz' dz \nn \\
 &=& S. \nn
\ea 


\end{proof}

\section{Quadratic reciprocity reformulated}
\label{secquadreciproc}

 We next consider a family of Hecke characters on $\bK$ parameterized by $0<s<1$: $\chi(x):=\chi_s\tilde{\chi}$, where $\chi_s(x):=|x|^s$, and the finite part $\tilde{\chi}$ is a quadratic character. For each place $\nu$ of the number field, one then has a complex scalar field theory as we have described in Section \ref{section2}, with action \eqref{complex action}, specified by $\chi_{\nu}$, the local component of $\chi$ at $\nu$. As a next step, we compute the two-point functions of these theories. In the following few paragraphs, we omit the subscript $\nu$ for simplicity of notations. To fix the normalizations, we pick the standard additive character and the self-dual Haar measure to define our actions.

A standard path integral calculation \cite{peskin} shows that the two-point function $G_0(x,y)$ is a Green's function for the Vladimirov derivative $D_\chi$ if the finite part of $\chi$ takes $-1$ to 1, and is a Green's function for $-D_{\chi}$ otherwise:

Consider the one-point function
\be
I=\int e^{-S}\phi(y) \DD\phi.
\ee Under a variation $\phi(x)\to \phi(x)+\delta \phi(x)$, $\overline{\phi}(x)\to \overline{\phi}(x)+ \delta\overline{\phi}(x)$, one has
\be
\delta S=\int \delta\overline{\phi}(x)D_\chi\phi(x)dx +\int \overline{\phi}(x)D_\chi\delta \phi(x)dx= \int  \delta\overline{\phi}(x)D_\chi\phi(x)dx +\int \overline{D}_\chi\overline{\phi}(x)\delta \phi(x)dx,
\ee 
and therefore
\be
\delta I= \int e^{-S} \lsb -\delta\overline{\phi}(x)D_\chi\phi(x)\phi(y)-\overline{D}_\chi\overline{\phi}(x)\delta \phi(x)\phi(y)+\delta\phi(x)\delta(x-y)\rsb dx\DD\phi.
\ee

On the other hand, we have $\delta I=0$ since we just performed an infinitesimal change of variable. $\delta\phi(x)$ and $\delta\overline{\phi}(x)$ are independent variations, as $\phi$ is a complex scalar field. Dividing by the partition function, we have
\be
\frac{\int e^{-S}D_\chi\phi(x)\phi(y)\DD\phi}{\int e^{-S}\DD\phi}=0,
\ee
and
\be
\frac{\int e^{-S}(\overline{D}_\chi\overline{\phi}(x)\phi(y)-\delta(x-y))\DD\phi}{\int e^{-S}\DD\phi}=0.
\ee
In particular, we have
\be
\overline{D}_\chi G_0(x,y)-\delta(x-y)=0,
\ee where
\be
G_0(x,y):=\frac{\int e^{-S}\overline{\phi}(x)\phi(y)\DD\phi}{\int e^{-S}\DD\phi}
\ee
is our two-point function. This is what we wanted, as one checks that $\overline{D}_{\chi}=D_{\chi}$ if $\tilde{\chi}(-1)=1$, and $\overline{D}_{\chi}=-D_{\chi}$ otherwise.

\begin{rmk}
As explained in \cite{Huang:2020aao}, the local functional equation in Tate's thesis then computes 
\be\label{G}
G_0(x,y)=\tilde{\chi}(-1)\mathcal F^{-1}\chi^{-1} = \tilde{\chi}(-1)\gamma(\chi) \chi\chi_1^{-1},
\ee 
where the gamma factor is $\gamma(\chi)=\epsilon(\chi,dx,\psi)L_{\nu}(\chi_1\chi^{-1})/L_{\nu}(\chi), $ and the epsilon factor $\epsilon(\chi,dx,\psi)$ depending on $\chi$, the local Haar measure $dx$, and the choice of the local additive character $\psi$, is equal to 1 almost everywhere. $L_{\nu}$ is the local $L$-factor at $\nu$. Moreover, the global functional equation implies that $G_0(x,y)$ satisfies the adelic product formula \be\label{G_0 product} \prod_{\nu\leq\infty}G_{0,\nu}(x,y)=1\ee
in the sense of analytic continuation, when $x\neq y\in\mathbb{K}$. (Note that $G_0$ depends on the place $\nu$, which was hidden in the discussions above.)
\end{rmk}

When the unitary part of $\tilde{\chi}(-1)=-1$,  $G_0(x,y)$ is purely imaginary. We multiply by $-i$ to make it real, i.e. we define $G'(x,y):=-iG_0(x,y)$ in this case. Otherwise, we define $G'(x,y):=G_0(x,y)$.

Finally, the local additive characters in Tate's thesis need to patch together so that the global additive character becomes trivial when restricted to the diagonal embedding of the number field. Effectively, in the case of $\bQ$, this means that the standard additive character takes the form $e^{-2\pi ix}$ at the real place, whereas at any $p$-adic place, there is no minus sign in front of the $2\pi i$. But when one writes down a physics theory independent of places, the choice of this sign is uniform across all places. So, one would like to flip the sign at the real place for the propagator. This sign flip will have no effect if the multiplicative at the real place is unramified. Otherwise, the sign flip shall flip the sign of $D_{\chi_{\nu}}$ there. Thus, we define the physical propagator to be $G(x,y)=-G'(x,y)$ at the real place, if the multiplicative character is ramified at the real place, and in any other cases, $G(x,y)=G'(x,y)$. One checks that $G(x,y)$ satisfies the adelic product formula $\prod_{\nu\leq\infty}G_{\nu}(x,y)=1$ 
in the sense of analytic continuation, as a consequence of Eq. \eqref{G_0 product} (again, $G$ has a hidden $\nu$ dependence).

Now, to connect with quadratic reciprocity, we specialize $\bK$ to be a quadratic extension of $\bQ$. In particular, we observe that a product identity of these two-point functions $G(x,y)$ is equivalent to the quadratic reciprocity law. We shall then move on to understand the identity in terms of the physics.

Specifically, given an odd prime $q$, consider the quadratic extension $\mathbb{K}=\mathbb{Q}(D)/\mathbb{Q}$, where $D:=(-1)^{(q-1)/2}q$. Among finite places, this extension is only ramified at $q$. It is ramified at $\infty$ (i.e. imaginary) iff $(q-1)/2$ is odd.

\begin{rmk}
We can consider other quadratic extensions. The above particular extension is the simplest one that is enough for our purposes here.
\end{rmk}

\begin{rmk}
 The closed string 4-tachyon genus zero scattering amplitude has a product formula in terms of such quadratic imaginary $\mathbb{K}$ \cite{FreundWitten}.
\end{rmk}


{\bf Main Observation:} There exists a unitary Hecke character $\tilde{\chi}$ of $\bQ$ (induced from a Dirichlet character), such that for any prime $p$ of $\mathbb{Q}$, any $x\neq y\in\mathbb{Q}$, and any $0<s<1$, one has
\be
\label{main}
\prod_{\nu \text{ above } p}G^\mathbb{K}_{\nu,\chi_s}(x,y) =  G^{\mathbb{Q}}_{p,\chi_s}(x,y)\frac{G^{\mathbb{Q}}_{p,\chi_s\tilde{\chi}}(x,y)}{\tilde{\chi}(x-y)},
\ee
where $G^\mathbb{K}_{\nu, \chi}$ is the propagator defined above associated to $D_\chi$ at the place $\nu$ of the number field $\mathbb{K}$. (By abuse of notation, a subscript $\chi$ on $G_\chi$ and $D_\chi$ really stands for the component of the Hecke character $\chi$ of $\mathbb{K}$ at $\nu$, i.e., $\chi_{\nu}$. The same notation $\chi_s$ on the left hand side denotes a Hecke character of $\bK$,  and on the right hand side denotes a Hecke character of $\bQ$.)

\begin{rmk}
Obviously, Eq. \eqref{main} holds for more general $s$. However, our physics interpretation concerns only $0<s<1$. 
\end{rmk}
\begin{rmk}\label{conductor}
By comparing the functional equations for the Dedekind zeta function of $\mathbb{K}$, and the functional equation for Dirichlet $L$-functions, and assuming that Eq. \eqref{main} holds at all places, one can deduce that $\tilde{\chi}$ has to be induced from the unique nontrivial quadratic Dirichlet character with conductor $q$ given by the Legendre symbol $\lb\frac{\cdot}{q}\rb$.
\end{rmk}

\begin{proof} (We also explain that Eq. \eqref{main} is equivalent to quadratic reciprocity.) We compute the propagators $G$ from $G_0$, which is given by Eq. \eqref{G}. We explain the equality separately for the Archimedean and non-Archimedean cases. For the non-Archimedean cases, there are three possibilities for an odd prime $p \subset \mathbb{Z}$ in the quadratic extension $\mathbb{K}$. The first possibility is that $p$ is splitting, that is $p = \pi_1 \pi_2$, where $\pi_1,\pi_2 \subset O_\mathbb{K}$, i.e. $\lb\frac{D}{p}\rb=1$. Since the Hecke character $\tilde{\chi}$ is induced from the Legendre symbol $\lb\frac{\cdot}{q}\rb$, and we have $\lb\frac{p}{q}\rb=\lb\frac{D}{p}\rb=1$ by quadratic reciprocity, then $(\chi_s\tilde{\chi})_p(p) = |p|^s_{\mathbb{Q}_p}$. I.e. $\tilde{\chi}$ is trivial at $p$. Since $|\pi_1|_{\mathbb{K}_{\pi_1}} =|\pi_2|_{\mathbb{K}_{\pi_2}} = |p|_{\mathbb{Q}_p} $, by Eq. \eqref{G} and the relation between $G$ and $G_0$, we conclude that Eq. \eqref{main} holds in this case. By a little more thought, one sees that \eqref{main} is indeed equivalent to the fact that $\lb\frac{p}{q}\rb=\lb\frac{D}{p}\rb$, i.e. the quadratic reciprocity for $p,q$ in this case.\newline
Next we discuss the case when $p$ is inert in $\mathbb{K}$, i.e. $\lb\frac{D}{p}\rb=-1$. As $p \neq q$, we have that the  Hecke character $\chi$ is still unramified at $p$, and $(\chi_s\tilde{\chi})_p(p) = - |p|^s_{\mathbb{Q}_p}$ as $\lb\frac{p}{q}\rb=\lb\frac{D}{p}\rb=-1$. Then by Eq. \eqref{G} and the relation between $G$ and $G_0$, the left hand side is given by 
\be
\frac{L(|\cdot|^{1-s}_{\mathbb{K}_p)}}{L(|\cdot|^{s}_{\mathbb{K}_p})} |x - y|^{s-1}_{\mathbb{K}_p} = \frac{1 - p^{-2s}}{1 - p^{2(s-1)}} |x -y|^{s-1}_{\mathbb{K}_p}.
\ee
The right hand side is given by
\be
\frac{L(|\cdot|^{1-s}_{\mathbb{Q}_p})}{L(|\cdot|^{s}_{\mathbb{Q}_p})} \frac{L(\hat{\chi}_p)}{L(\chi_p)} |x - y|^{2(s-1)}_{\mathbb{Q}_p} = \frac{1 - p^{-s}}{1 - p^{s-1}} \frac{1  + p^{-s}}{1 + p^{s-1}} |x -y|^{2(s-1)}_{\mathbb{Q}_p},
\ee 
where $\hat{\chi}_p:={(\chi_1)}_p\chi_p^{-1}$. Thus the equality holds. Again, one sees that Eq. \eqref{main} is equivalent to the fact that $\lb\frac{p}{q}\rb=\lb\frac{D}{p}\rb$, i.e. the quadratic reciprocity for $p,q$ in this case.\newline
For the case that $p = q$, i.e. $p$ is ramified, we have that the unitary part of the Hecke character at this place is equal to the Legendre symbol $\lb\frac{\cdot}{q}\rb$ with conductor $1$. On the other hand, the standard additive character on $\mathbb{K}_\pi$ is $\psi_\pi = e^{2\pi i \Tr(x)}$ where $\pi = \sqrt{p}$. This additive character also has conductor $1$. Then the left-hand side gives 
\bbe
q^{-(s - \frac{1}{2})} \frac{L(|\cdot|^{1-s}_{\mathbb{K}_\pi})}{L(|\cdot|^s_{\mathbb{K}_\pi})} |x - y |^{s-1}_{\mathbb{K}_\pi},
\ee 
where the factor $q^{-(s-\frac{1}{2})}$ comes from the epsilon factor associated with the dual pair $(\psi_\pi ,dx)$ and the unramified multiplicative character, also $q = |\pi|_{\mathbb{K}_\pi}$. The right-hand side gives 
\ba
& &\begin{cases}
-i\frac{L(|\cdot|^{1-s}_{\mathbb{Q}_p})}{L(|\cdot|^{s}_{\mathbb{Q}_p})} q^{-s}\big\{\sum^{q-1}_{j = 1}\tilde{\chi}_p(j) e^{j2 \pi \frac{i}{p} } \big \}|x - y|^{2(s-1)}_{\mathbb{Q}_p} & \frac{q-1}{2} \text{ is odd} \\
\frac{L(|\cdot|^{1-s}_{\mathbb{Q}_p})}{L(|\cdot|^{s}_{\mathbb{Q}_p})} q^{-s}\big\{\sum^{q-1}_{j = 1}\tilde{\chi}_p(j) e^{j2 \pi \frac{i}{p} } \big \}|x - y|^{2(s-1)}_{\mathbb{Q}_p} & \frac{q-1}{2} \text{ is even}
\end{cases} \\
& &=\begin{cases} 
    (-i) \cdot i\frac{L(|\cdot|^{1-s}_{\mathbb{Q}_p})}{L(|\cdot|^{s}_{\mathbb{Q}_p})} q^{-s+ \frac{1}{2}}|x - y|^{2(s-1)}_{\mathbb{Q}_p} & \frac{q-1}{2} \text{ is odd} \\
      \frac{L(|\cdot|^{1-s}_{\mathbb{Q}_p})}{L(|\cdot|^{s}_{\mathbb{Q}_p})} q^{-s+ \frac{1}{2}}|x - y|^{2(s-1)}_{\mathbb{Q}_p} & \frac{q-1}{2} \text{ is even}
   \end{cases}. \nn
\ea 
Notice that the last equality holds by the property of the quadratic Gaussian sum. The extra $-i$ for the case $\frac{q-1}{2}$ is odd comes from the fact that in this case $\tilde{\chi}_p$ takes $-1$ to $-1$ and therefore  $G(x,y) = - i G_0(x,y)$.

For the Archimedean case, if $\frac{q- 1}{2}$ is even, then there are two different real embeddings of $\mathbb{K}$. The character on the right-hand side is also unramified, therefore the equality holds trivially. For the case that $\frac{q-1}{2}$ is odd, we have that the left hand side is a complex embedding. Therefore, the left-hand side equals 
\bbe
\frac{(2\pi)^{-(1-s)} \Gamma(1-s)}{(2\pi)^{-s} \Gamma(s) } |x - y|^{s-1}_\mathbb{C},
\ee 
while the right-hand side is 
\bbe
i \cdot (-i) \frac{\pi^{-(\frac{2-s}{2})}\Gamma(\frac{2-s}{2})}{\pi^{-\frac{s+1}{2}} \Gamma(\frac{s+1}{2})} \frac{\pi^{-\frac{1-s}{2}} \Gamma(\frac{1-s}{2})}{\pi^{-\frac{s}{2}} \Gamma(\frac{s}{2})} |x - y|^{2(s-1)}_{\mathbb{R}}.
\ee 
From the Legendre duplication formula, we know that left-hand side and right-hand side agree.
\end{proof}

\begin{rmk}
 Take $\nu$ to be the real place, and take the derivative at $s=1$, one recovers the familiar (obvious) relation between closed string (left-hand side) and open string (right-hand side) propagators.
\end{rmk}
 
 \begin{rmk}[A simple recast of the above computations] The local gamma factor is a ratio of local $L$-factors, times an epsilon factor. These local $L$-factors for $\mathbb{K}$ have a factorization in terms of the $L$-factors for $\mathbb{Q}$ of the Galois characters of $\mrm{Gal}(\mathbb{K}/\mathbb{Q})$. After taking care of the epsilon factors, at a finite odd place $p\neq q$, identity \eqref{main} therefore is equivalent to the quadratic reciprocity
\be
\lb\frac{D}{p}\rb= \lb\frac{p}{q}\rb.
\ee
The left hand side comes from the nontrivial Galois character, and the right hand side comes from the nontrivial Dirichlet character.
\end{rmk}
 
 \begin{rmk}
 For the case $p=2$, Hensel's lemma works differently, and one can check that the main observation is equivalent to the supplemental law $\lb\frac{2}{q}\rb=(-1)^{\frac{q^2-1}{8}}.$
 \end{rmk}

\section{Quadratic reciprocity and adelic conformal field theories}
\label{factorization}

We now present the physics interpretation of quadratic reciprocity in terms of adelic conformal field theories, by explaining that Eq. \eqref{main} follows from a proposed holomorphic factorization of the above adelic conformal field theories.

The new idea here is that at each place, such product formulas arise from an adelic holomorphic factorization property of our generalized free conformal field theories. Furthermore, the fact that these product formulas hold at every place in a globally consistent way, in the sense that they come from restricting adelic conformal field theories at local places, is equivalent to the reciprocity law.

More specifically, specializing at a complex Archimedean place $\bC$, the field theory that has been discussed, is a generalized free field theory, whereas the propagator is not the standard free field propagator. From its transformation law under the global conformal group, it is clear that it is a quasi-primary field of conformal weight $(\frac{1-s}{2}, \frac{1-s}{2} )$. Next, one assumes that $\phi(z,\bar{z})$ has a holomorphic-anti-holomorphic factorization
\be\label{factor}
 \phi(z,\bar{z})=\phi(z)\bar{\phi}(\bar{z}),
 \ee
where $\phi(z)$ and $\bar{\phi}(\bar{z})$ are holomorphic (chiral), and anti-holomorphic (anti-chiral) quasi-primary fields, necessarily of conformal weights $(\frac{1-s}{2}, 0)$ and $(0, \frac{1-s}{2} )$. 
i.e. their transformation laws under $g\in \mrm{GL}(2,\bC)$ are

\be\label{hol}
 \phi(z)\cdot g = \phi(g(z)) (cz + d)^{s-1} \det(g)^{(1-s)/2}
\ee 
and 
\be\label{anti-hol}
\bar{\phi}(\bar{z})\cdot g = \bar{\phi}((\overline{g(z)})) \overline{(cz + d)}^{s-1} \overline{\det(g)}^{(1-s)/2}.
\ee

One checks that the above transformation laws and the factorization are compatible with the transformation law of $\phi(z,\bar{z})$ under the global conformal group. $\phi(z)$ and $\bar{\phi}(\bar{z})$ are generally multi-valued fields unless $s$ takes special integer values.

The two-point function of $\phi(z,\bar{z})$ is determined by the most singular term of the OPE of $\phi(x,\bar{x})$ and $\phi(y,\bar{y})$, which then is the product of the OPEs of $\phi(z)$ and $\bar{\phi}(\bar{z})$.\footnote{There is a subtlety here regarding the OPE, which is related to the extra factor of $-i$ one needs to take care of, as explained in the previous section. It shall be explained in Remark \ref{-i}.} The multi-valuedness of these fields does not affect the OPE of their products, as long as one always chooses their branches in a way so that their product is rotationally invariant. Now, restrict the chiral field $\phi(z)$ to $\bR^+$. We next look at how it transforms under elements of the real global conformal group, that preserve $\bR^+$, and also preserve the ordering on $\bR$. For any $g\in \mrm{GL}(2,\bR)$ that preserves $\bR^+$, one checks that $g$ has to have all nonnegative entries, or all nonpositive entries. Since the center of $\mrm{GL}(2,\bR)$ acts trivially on the quasi-primary fields, we can assume without loss of generality that $g$ has all nonnegative entries. Among these group elements, the additional requirement that it   preserves the ordering on the real line under the fractional linear transformation, is equivalent to the condition that $\det(g)>0$. One checks that under such group elements $g$, the transformation of $\phi(z)$ restricted to $\bR^+$ matches that of a real quasi-primary field of conformal weight $\frac{1-s}{2}$. 

Such a real quasi-primary field has only two possibilities: its transformation law is either according to the unramified real character of exponent $\frac{s-1}{2}$, i.e. $|\cdot|^{\frac{s-1}{2}}$, or via the ramified real character $|\cdot|^{\frac{s-1}{2}}\sgn(\cdot)$, as these are all the possible multiplicative characters of $\bR^\times$ with the correct weight. (Note that restricting to $\bR^+$, the sign character becomes trivial, and these two real quasi-primary fields transform in the same way under the above specified group elements.) The same argument applies to the anti-chiral field. So the chiral and anti-chiral fields must be extended from these real quasi-primary fields. Since they are chiral or anti-chiral, they are also uniquely determined by their restrictions to $\bR^+$. On the other hand, they cannot restrict to the same real quasi-primary field on $\bR^+$, since that would force them to have identical coefficients in their Laurent expansions in $z$ and $\bar{z}$, respectively. However, this is not possible since these coefficients represent different physical degrees of freedom. Therefore, the chiral and anti-chiral fields restrict to exactly the two possible real quasi-primary fields of conformal weight $\frac{1-s}{2}$, and thus their OPEs restricting to $\bR^+$ are identical of the OPEs of these two real quasi-primary fields, respectively, which in turn are equal to the 2-point functions of these real quasi-primary fields. Thus, this explains the product formula of the two-point functions \eqref{main} at Archimedean places.

\begin{rmk}\label{-i}
Here is an additional detail in computing the OPE of $\phi(z,\bar{z})$ restricted to the real line, using the factorization $\phi(z,\bar{z})=\phi(z)\bar{\phi}(\bar{z})$ restricted to the real line. To take care of the ordering, we put $x>0>y$, and compute $\overline{\phi(x,\bar{x})}\phi(y,\bar{y})$. However, on the negative real line, the factorization becomes instead $\phi(z,\bar{z})=-i\phi(z)\bar{\phi}(\bar{z})$, as one can apply the group element \be g =\begin{bmatrix}
-1 & 0 \\
0 & 1 
\end{bmatrix}
\ee  to the factorization on the positive real line to see it, given how the two real quasi-primary fields of weight $\frac{1-s}{2}$ obtained by the restriction transform under $g$. (To be more precise, there is a choice of a square root of   $\tilde{\chi}(-1)=-1$ involved, as one sees from the transformation law \eqref{action}. This choice is determined by the choice of the overall sign for the additive character: the standard character, or negative of the standard character.) The extra factor of $-i$ comes from the nontrivial central term. This explains the extra factor of $-i$ in the Green's function factorization at the Archimedean place, which is due to a nontrivial central term caused by $\tilde{\chi}(-1)=-1$.
\end{rmk}

We expect that this reasoning extends to the $p$-adic places that are either inert or ramified (the split case is trivial), where the factorization to holomorphic and anti-holomorphic pieces of the quasi-primary field on $\mathbb{K}_{\nu}$, the local field extension, is generalized to a factorization into a product of fields depending on the conjugates of the variable $z\in \mathbb{K}_{\nu}$ under the local Galois group. More specifically, since the physics action on $\mathbb{K}_{\nu}$ comes from a Hecke character with a trivial finite part, at every place, the action is clearly invariant under the action of the local Galois group. This suggests a $p$-adic version of \eqref{factor}, where $z, \bar{z}$ are ``holomorphic'' and ``anti-holomorphic'' coordinates on the local field extension, and the actions of the global conformal group $G$ on the "chiral" and ``anti-chiral'' quasi-primary fields, are generalizations of \eqref{hol} and \eqref{anti-hol}. 

Next, the global compatibility of these local product formulas of the two-point functions is equivalent to the reciprocity law. More specifically, given a continuous character on the idele group of the number field $\mathbb{K}$, one can associate a conformal field theory at each place as described, using the local component of the global character. Then, the physics requirement that the local two-point functions multiplying to 1 requires the global character to factorize through the idele class group, i.e. the global character has to be a Hecke character. In the product formula \eqref{main}, there are 3 global conformal field theories involved, one on $\mathbb{K}$ with trivial Hecke character, one on $\bQ$ with the trivial Hecke character, and the other one on $\bQ$ with a non-trivial Hecke character, such that at each place, the local conformal field theories appearing in the product formula come from the global characters as described above. Mathematically, then, the globally compatible product formulas imply that the Dedekind zeta function of $\mathbb{K}$ equals the product of the Dedekind zeta function of $\bQ$, i.e Riemann zeta function, together with the $L$-function of the non-trivial Hecke character on $\bQ$. One can furthermore show that the conductor of the non-trivial Hecke character is $q$, by looking at the functional equation of the Dedekind zeta function of $\mathbb{K}$, as is explained in Remark \ref{conductor}. Therefore, this non-trivial Hecke character has to be induced from the unique non-trivial quadratic Dirichlet character on $(\bZ/q\bZ)^\times$ given by the quadratic symbol. This is exactly the analytic formulation of the quadratic reciprocity law.

\subsection{Physics derivation of the holomorphic factorization in the Archimedean case}

In this section, we show that the holomorphic factorization of the $\phi$ field holds in the Archimedean case. We expect a similar reasoning at the non-Archimedean places, but we defer the treatment of non-Archimedean places to a future work.

In the Archimedean case, when $s=1$, one has a holomorphic copy of the Virasoro algebra, and an anti-holomorphic copy of the Virasoro algebra, acting as symmetries of the 2d CFTs. Furthermore, their actions commute. In other words, one has a symmetry action of $\text{Vir}\oplus \text{Vir}$ on the CFTs. This fact is the origin of the phenomenon of holomorphic factorization in this Archimedean 2d CFT. 

Next we discuss the holomorphic factorization, for $0<s<1$. We show that in this case, although the Witt algebra of infinitesimal conformal transformations no longer preserves the action, the action is still invariant under the global conformal group. Specifically, we show that the action is invariant under two commuting copies (holomorphic and anti-holomorphic) of the Lie algebra $sl(2,\mathbb C)$.

\begin{lem}\label{com}
Let $s$ be a positive real number, then \begin{equation}
\lsb D_s,z\rsb=\lb\frac{1}{2\pi i}\rb^2 sD_{s-1}\partial_{\bar{z}}.
\end{equation}
\end{lem}

\begin{proof}
When $s$ is a positive integer, the identity is familiar. For generic real $s>0$, we prove the above identity by using the local Fourier transform:
\begin{equation}
     \FF^{-1}\lsb D_s,z\rsb\FF=\lsb (z\bar{z})^s,\frac{1}{2\pi i}\partial_z\rsb=\frac{-1}{2\pi i}s(z\bar{z})^{s-1}\bar{z}.
\end{equation}
So
\begin{equation}
\lsb D_s,z\rsb = \frac{-1}{2\pi i}s\FF s(z\bar{z})^{s-1}\FF^{-1}\FF \bar{z}\FF^{-1}=\lb\frac{1}{2\pi i}\rb^2 sD_{s-1}\partial_{\bar{z}}.
\end{equation}
\end{proof}

\begin{rmk}
Note that $D_s$ is defined for all complex $s$ except at poles of the zeta integral. We shall also need the identity
\begin{equation}
    [D_s,z^2]=\lb\frac{1}{2\pi i}\rb^2\lsb 2szD_{s-1}\partial_{\bar{z}}-\frac{1}{2\pi i} s(s-1)D_{s-2}\partial_{\bar{z}}^2 \rsb,
\end{equation}
which follows directly from the above lemma.
\end{rmk}

\begin{thm}\label{CFT}
The physics action $S$ at the complex place is invariant under an action of the holomorphic $sl(2,\mathbb C)$ given by $Y=-\partial_z, X=z^2\partial_z+2hz, H=2z\partial_z+2h$, where $h=\frac{1-s}{2}$ is the holomorphic conformal weight of the quasi-primary field $\phi$. Likewise, the physics action is also invariant under an action of the anti-holomorphic copy of $sl(2,\mathbb C)$, given by the complex conjugate. These two actions commute.
\end{thm}

\begin{proof}
We only prove the statement regarding the action of the holomorphic copy of the Lie algebra $sl(2,\mathbb C)$. The other statements will then be clear.

First, it is clear that the above expressions of $X,Y,H$ satisfy the commutation relations of the standard generators of $sl(2,\mathbb C)$: i.e. $[H,X]=2X, [H,Y]=-2Y, [X,Y]=H$. 

Next, consider the variation of the action, under an infinitesimal action on the quasi-primary field $\phi$, generated by $-Y=\partial_z$. i.e. $\phi\to \phi+\epsilon\partial_z\phi$. We would like to show that the corresponding variation of the physics action $S$ vanishes to first order in $\epsilon$. We have

\begin{equation}
    \frac{\delta S}{\epsilon}=\int \partial_z\bar{\phi}D_s\phi+ \int \bar{\phi}D_s\partial_z\phi=\int \bar{\phi}(-1)\partial_zD_s\phi+\int \bar{\phi}D_s\partial_z\phi=0.
\end{equation}

In the above, one uses that the operator $\partial_z$ is anti self-adjoint, and that $D_s$ and $\partial_z$ commute, since they are both given by Fourier conjugates of quasi-characters.

Next, consider the same variational problem w.r.t. $H/2=z\partial_z+h$. We have
\begin{equation}
 \frac{\delta S}{\epsilon}=\int (z\partial_z+h)\bar{\phi}D_s\phi+\bar{\phi}D_s(z\partial_z+h)\phi=\int(-1)\bar{\phi}\partial_zzD_s\phi+\int \bar{\phi}D_sz\partial_z\phi+2h\int \bar{\phi}D_s\phi.   
\end{equation} 
One sees that the right-hand side equals 0 by Lemma \ref{com}.

Finally, we consider the same variational problem w.r.t. $X=z^2\partial_z+2hz$. We have
\begin{align}
    \frac{\delta S}{\epsilon} 
    &= \int (z^2\partial_z+2hz)\bar{\phi}D_s\phi+\int \bar{\phi}D_s(z^2\partial_z+2hz)\phi \\
    &= \int 2hz\bar{\phi}D_s\phi+\int (-1)\bar{\phi}\partial_zz^2D_s\phi+\int \bar{\phi}D_sz^2\partial_z\phi+\int 2hz\bar{\phi}D_s\phi \\
    &= \int 2hz\bar{\phi}D_s\phi-\int \bar{\phi}z^2\partial_zD_s\phi-\int \bar{\phi}2zD_s\phi+\int \bar{\phi}z^2D_s\partial_z\phi+2s\int z\bar{\phi}D_s\phi\\
    &\ + s(s-1)\lb \frac{1}{2\pi i}\rb^2\int \bar{\phi}D_{s-1}\partial_{\bar{z}}\phi+h\int \bar{\phi}zD_s\phi+hs\lb \frac{1}{2\pi i}\rb^2\int \bar{\phi}D_{s-1}\partial_{\bar{z}}\phi \nn\\
    &= 0.
\end{align}

In the above we have used Lemma \ref{com}, the identity in the remark below the lemma, and that $h=(1-s)/2$.
\end{proof}

\begin{rmk}
One could try to see if the action is also invariant under e.g. an operator of the form $z^3\partial_z+\lambda z^2$. It turns out that certain terms would not cancel unless $s=1$, i.e. when $s=1$, the physics action becomes local, and it enjoys the much larger symmetry of the local conformal algebra. For generic $s$, the symmetry shrinks to the above two copies of $sl(2,\mathbb C)$.
\end{rmk}

Now, our quasi-primary field $\phi(z,\bar{z})$ with weight $(h,h)$ has an expansion 

\begin{equation}
    \phi(z,\bar{z})=\sum_{i,j\in\mathbb Z} \frac{a_{ij}}{z^{i+h}\bar{z}^{j+h}},
\end{equation}
where the coefficients $a_{ij}$ are operators. We next work out the action of our $sl(2,\mathbb C)\oplus sl(2,\mathbb C)$ on the vector space spanned by these coefficients, according to the rule that the assignment from a quasi-primary field to the $ij$-th coefficient is equivariant under the symmetries of the physics action. 

First, for a holomorphic quasi-primary field of conformal weight $h$, one can expand it as
\be
\phi(z)=\sum_{i\in Z} \frac{a_i}{z^{i+h}}.
\ee

We only need to work out the action of the Lie algebra on $z^{-k-h}$, for $k\in\mathbb Z$. We have $Hz^{-k-h}=-2kz^{-k-h}$, $Xz^{-k-h}=(-k+h)z^{1-k-h}$, $Yz^{-k-h}=(k+h)z^{-1-k-h}$. 

From this we read off that the the holomorphic $sl(2,\mathbb C)$ acts on the coefficients as the following representation $V: V=\oplus_{n\in \mathbb Z} V_n$, where each $V_n$ is a 1-dimensional eigenspace under $H$ with eigenvalue $-2n$, and the Casimir operator acts by the scalar multiple $4h(h-1)$. Each $V_n$ is spanned by the coefficient $a_n$. Likewise, we have another copy of this representation for the anti-holomorphic quasi-primary field 
\be
\bar{\phi}(\bar{z})=\sum_{i\in Z} \frac{\bar{a}_i}{\bar{z}^{i+h}}.
\ee

Note in particular, when $h$ is not an integer, any given coefficient of such a holomorphic quasi-primary field determines the field, as one can reach any coefficient from any given coefficient by applying the lowering and raising operators $X$ and $Y$. This in turn implies that any coefficient can not be zero, unless the field is zero. So the representation is irreducible.

Now we see that the above structure implies the desired holomorphic factorization of our quasi-primary field $\phi(z,\bar{z})$: from the above expansion of $\phi(z,\bar{z})$, we see that the coefficients span a representation of $sl(2,\mathbb C)\oplus sl(2,\mathbb C)$, that is an external tensor product of the above two irreducible representations of $sl(2,\mathbb C)$. In particular, there is only one vector in the representation with eigenvalue $(0,0)$. This implies the factorization $a_{00}=a_0\otimes \bar{a}_0$. Since $a_{00}$ determines $\phi(z,\bar{z})$ according to the tensor product representation, the factorization of $a_{00}$ therefore forces the holomorphic factorization of $\phi(z,\bar{z})$, i.e. $\phi(z,\bar{z})=\phi(z)\bar{\phi}(\bar{z})$.

\subsubsection{Relation to dimensional regularization}
The deformation of the free scalar field that we have been discussing is in fact closely related to dimensional regularization in the following way: when a Feynman graph has a UV divergence, dimensional regularization is often described by changing the dimension of the spacetime from an integer to a complex number, so that the volume of a unit ball in momentum space becomes a meromorphic function of the dimension, with a pole at the integral spacetime dimension that one is interested in. Dimensional regularization then proceeds by first computing the Feynman integral when the real part of the complex dimension is sufficiently small, and then by a description of how to extract a finite quantity from the Feynman graph at this pole by analytic continuation. This procedure has been made mathematically rigorous by Etingof in \cite{Etingof}, where, instead of modifying the spacetime dimension, one modifies the propagator by raising it to a power parameterized by a complex number, thus the propagator becomes a distribution valued meromorphic function of this parameter. One then proves that this function has a meromorphic continuation to the whole complex plane, using e.g. the existence of the Bernstein-Sato polynomial of the Fourier transform of the d'Alembert operator (or Laplacian operator in Euclidean signature). This meromorphic continuation then provides distributional inverses of the Fourier transform of the d'Alembert operator, which are the correct propagators that one should insert into the Feynman integrals. The observation here is that, in an Euclidean scalar field theory, after we deform the theory by modifying the Laplacian by raising it to the power $s$, the propagator of the deformed theory in momentum space is equal to the $s$-th power of the original propagator, i.e. it is exactly the modified propagator that one uses in the rigorous treatment of dimensional regularization. Therefore, dimensional regularization in fact makes use of this family of non-local deformations: for a generic $s$, the momentum space propagator is a well-defined distribution, while at the pole $s=1$, the propagator can be obtained by the meromorphic continuation of the propagator of the deformed family of theories.

\appendix
\section{The main observation in the case of one and two dimensions}

\noindent In this appendix we will explain how the main observation applies in the case of one-dimensional real boson and fermion theories, versus a two-dimensional real boson. This physical setup is slightly different from the one considered in the main body of the text, however the main observation still holds, as the mathematical content is equivalent. In order to illustrate the result for readers who may not be familiar with Tate's thesis, we will directly calculate the 2-point functions by the path integrals.

Our setup is as follows. We consider two one-dimensional generalized free theories, of a scalar boson and fermion respectively, as well as the theory of a free boson on the complex plane. We will explicitly check that the Green's functions for these theories obey the main observation in Eq.~\eqref{main}.

\subsection{One-dimensional boson}

\begin{rmk}[Notation and conventions]
In this appendix only we will denote the quasi-character at the Archimedean places $\mathbb{R}$ and $\mathbb{C}$ by $\pi$, and furthermore we will compute the path integrals by integrating against $e^{iS}$ (where $S$ is the action). The Fourier transform of quasi-character $\pi$ is 
\be
\int_\mathbb{R} \pi(x)e^{2\pi i k x} dx = \frac{\Gamma\lb \pi |\cdot| \rb}{\pi(k)|k|},
\ee
and the values for the Gelfand-Graev Gamma function are
\ba
\Gamma(|\cdot|^s) &=& 2^{1-s} \pi ^{-s} \cos \left(\frac{\pi  s}{2}\right) \Gamma_E(s), \\
\Gamma(|\cdot|^s \sgn_{-1}(\cdot) ) &=& i 2^{1-s} \pi ^{-s} \sin \left(\frac{\pi  s}{2}\right)\Gamma_E (s),
\ea
where $\Gamma_E(s)$ is the Euler Gamma function.
\end{rmk}

\begin{lem}
For a scalar field $\phi:\mathbb{R}\to\mathbb{R}$ the Fourier coefficients obey
\be
\bar{c}_{-k}=c_k.
\ee
\end{lem}
\begin{proof}
We expand $\phi$ in Fourier modes as
\be
\phi(x) = \int_\mathbb{R} c_k e^{2\pi i k x} dk,
\ee
so that
\ba
\bar\phi(x) &=& \int_\mathbb{R} \bar{c}_k e^{-2\pi i k x} dk\\
&=& \int_\mathbb{R} \bar{c}_{-k} e^{2\pi i k x} dk
\ea 
under a change of variables. Then the reality condition $\bar{\phi}(x)=\phi(x)$ implies
\be
\bar{c}_{-k} = c_k.
\ee
\end{proof}

\begin{rmk}
As explained in the main text of the paper, the Vladimirov derivative associated to quasi-character $|\cdot|^s$ acts on a Fourier mode~as
\ba
D_s e^{2\pi i k x} &\coloneqq& |k|^s e^{2\pi i k x} \\
&\eqqcolon& \lambda_s(k) e^{2\pi i k x},
\ea
\end{rmk}

\begin{defn}
The theory of a free scalar boson has action
\be
S_b \coloneqq \int \frac{1}{2} dx \phi(x) D_s\phi(x),
\ee
and the partition function is
\be
I \coloneqq \int \DD\phi e^{i \int \frac{1}{2} dx \phi D_s\phi }.
\ee
\end{defn}

\begin{lem}
\label{lmmaA1}
The 2-point function $\langle \phi(x_1) \phi(x_2) \rangle_b$ equals
\be
\langle \phi(x_1) \phi(x_2) \rangle_b = 2 i \Gamma\lb |\cdot|^{1-s} \rb |x_2-x_1|^{s-1}.
\ee
\end{lem}
\begin{proof}
We first compute $I$ formally by plugging in the Fourier decomposition. We have
\ba
I &=& \int \DD c \exp\lsb \frac{i}{2} \int dx dk dk' \lb c_k e^{2\pi i k x} \rb \lb c_{k'} \lambda_s(k') e^{2\pi i k'x} \rb \rsb \\
  &=&  \int \DD c \exp \lsb \frac{i}{2} \int dkdk' c_k c_{k'} \lambda_s(k') \delta(k+k') \rsb \\
  &=&  \int \DD c \exp \lsb \frac{i}{2} \int dk c_{-k} c_{k} \lambda_s(k) \rsb \\
  &=&  \int \DD c \exp \lsb \frac{i}{2} \int dk |c_{k}|^2 \lambda_s(k) \rsb \\
  &=& \prod_k \sqrt{\frac{2i\pi}{\lambda_s(k)}}.
\ea
We now compute the 2-point function,
\ba
\langle \phi(x_1) \phi(x_2) \rangle_b &=& \frac{1}{I} \int \DD\phi \phi(x_1)\phi(x_2) e^{iS_b} \\
&=& \frac{1}{I} \int \DD c dk_1 dk_2 c_{k_1} c_{k_2} e^{2\pi i \lb k_1 x_1 + k_2 x_2\rb} \exp \lsb \frac{i}{2} \int dk |c_{k}|^2 \lambda_s(k) \rsb \\
\label{eq121}
&=& \frac{2}{I} \int \DD c dk_1 |c_{k_1}|^2 e^{2\pi i (x_1-x_2)k_1} \exp \lsb \frac{i}{2} \int dk |c_{k}|^2 \lambda_s(k) \rsb \\
&=& \int dk_1 e^{2\pi i \lb x_1-x_2 \rb k_1} \frac{2i}{\lambda_s(k_1)}\\
&=& 2 i \Gamma\lb |\cdot|^{1-s} \rb |x_2-x_1|^{s-1},
\ea
where in step \eqref{eq121} we have used that there are two contractions contributing to the 2-point function.
\end{proof}

\subsection{One-dimensional fermion} The fermionic action is
\be
S_f \coloneqq \frac{i}{2} \int dx \phi(x) D_{s,-} \phi(x),
\ee
where the factor of $i$ ensures the action is real and the derivative associated to quasi-character $|\cdot|^s\sgn_{-1}(\cdot)$ acts on a Fourier mode as
\ba
D_{s,-} e^{2\pi i k x} &\coloneqq& |k|^s \sgn{(k)} e^{2\pi i k x} \\
&\eqqcolon& \lambda_{s,-}(k) e^{2\pi i k x}.
\ea

\begin{lem}
The action $S_f$ vanishes if field $\phi$ is bosonic. It does not vanish if $\phi$ is fermionic.
\end{lem}

\begin{proof}
We have
\ba
S_f &=& \frac{i}{2} \int dx dk dk' \lb c_{k'} e^{2\pi i k' x} \rb \lb c_{k} |k|^s \sgn{(k)} e^{2\pi i kx} \rb \\
&=& \frac{i}{2}\int dk c_{-k} c_{k} |k|^s \sgn{k}.
\ea
Under the change of variables $k\to-k$ this equals
\be
S_f = - \frac{i}{2} \int dk c_{k} c_{-k} |k|^s \sgn{k} \\
\ee
The (anti-)commutation relations for the mode creation/annihilation operators $c_k$ of the field $\phi$ are
\be
\bar c_k c_k = (-1)^F c_k \bar c_k,
\ee
where $F=1$ if $\phi$ is fermionic and $0$ if $\phi$ is bosonic, and we have dropped the additive constant. Using that $c_{-k}=\bar c_k$, the action $S_f$ thus becomes
\ba
S_f &=& -(-1)^F \frac{i}{2} \int dk c_{-k} c_{k} |k|^s \sgn{k} \\
    &=& -(-1)^F S_f,
\ea
which proves the lemma.
\end{proof}

\begin{rmk}
A computation analogous to the one in Lemma \ref{lmmaA1} shows the 2-point function is
\ba
\langle \phi(x_1) \phi(x_2) \rangle_f &=& \int dk e^{2\pi i \lb x_1-x_2 \rb k_1} \frac{1}{|k|^s\sgn k}\\
&=& \Gamma\lb |\cdot|^{1-s} \sgn_{-1}(\cdot) \rb |x_1-x_2|^{s-1} \sgn_{-1}\lb x_1 - x_2 \rb.
\ea
\end{rmk}

\subsection{Bosonic field on $\mathbb{C}$}
The action in this case is given by
\be
S_{b,\mathbb{C}} \coloneqq \frac{1}{2} \int dz \phi(z) \phi D_s \phi(z), 
\ee
where the Vladimirov derivative $D_s$ associated to quasi-character $|\cdot|^s$ acts on the Fourier mode $e^{2\pi i \Re\lb 2kz \rb}$ as
\ba
D_s e^{2\pi i \Re\lb 2kz \rb} &\coloneqq& |k|^s e^{2\pi i \Re\lb 2kz \rb} \\
&\eqqcolon& \lambda_s(k) e^{2\pi i \Re\lb 2kz \rb},
\ea
and the absolute value is defined by $|k|=\sqrt{ \lb \Re k \rb^2 + \lb  \Im k \rb^2 }$ as usual, which is the square root of the convention used in the main text. The Fourier coefficients in this case can be obtained as
\be
c(k) = \int \phi(z) e^{-2\pi i \Re\lb 2kz \rb}dz,
\ee
so that 
\be
\phi(z) = \int c(k) e^{2\pi i \Re\lb 2kz \rb} dk,
\ee
and the reality of $\phi(z)$ implies
\be
\bar{c}(k) = c(-k).
\ee

Now consider the partition function path integral and the two-point function. We have the following lemma
\begin{lem}
The 2-point function $\langle \phi(z_1) \phi(z_2) \rangle_{b,\mathbb{C}}$ equals
\be
\langle \phi(z_1) \phi(z_2) \rangle_{b,\mathbb{C}} = 4 \Gamma \lb|\cdot|^{1-\frac{s}{2}}\rb \Gamma \lb|\cdot|^{1-\frac{s}{2}}\sgn_{-1}(\cdot)\rb |z_1-z_2|^{s-2}.
\ee
\end{lem}
\begin{proof}
We have
\ba
I &=& \int \DD \phi \exp\lsb \frac{i}{2} \int dz \phi D_s \phi \rsb \\
&=& \int \DD c \exp\lsb \frac{i}{2} \int dz dk dk' \lb c(k') e^{2\pi i \Re\lb 2k'z \rb} \rb \lb c(k) \lambda_s(k) e^{2\pi i \Re\lb 2kz \rb} \rb \rsb \\
  &=&  \int \DD c \exp \lsb \frac{i}{2} \int dkdk' c(k') c(k) \lambda_s(k) \delta\lb 2 \Re\lb k+k'\rb\rb \delta\lb-2\Im\lb k+k'\rb\rb \rsb \\
  &=&  \int \DD c \exp \lsb \frac{i}{8} \int dk |c(k)|^2 \lambda_s(k) \rsb \\
  &=& \prod_{k} \sqrt{\frac{8\pi}{i\lambda_s(k)}},
\ea
where the product is over pairs in the complex plane.

For the two-point function,
\ba
\langle \phi(z_1) \phi(z_2)}\rangle_{b,\mathbb{C} &=& \frac{1}{I} \int D\phi \phi(z_1)\phi(z_2) e^{iS_{b,\mathbb{C}}} \\
&=& \frac{1}{I} \int Dc dk dl c(k) c(l) e^{2\pi i \Re\lb 2kz_1 + 2lz_2 \rb} \exp \lsb \frac{i}{8} \int dp |c(p)|^2 \lambda_s(p) \rsb \nn \\
&=& \frac{2}{I} \int Dc dk |c(k)|^2 e^{2\pi i \Re\lb 2k(z_1 - z_2) \rb} \exp \lsb \frac{i}{8} \int dp |c(p)|^2 \lambda_s(p) \rsb \\
&=& \int dk e^{2\pi i \Re\lb 2k(z_1 - z_2) \rb} \frac{8i}{\lambda_s(k^i)}\\
&=& \int dk e^{2\pi i \lsb k_1(z_1 - z_2)_1 + k_2(z_1 - z_2)_2  \rsb} \frac{2i}{|k/2|^s}.
\ea
We now parameterize the plane integrals as
\ba
\langle \phi(z_1) \phi(z_2) \rangle_{b,\mathbb{C}} &=& \int_0^\pi \int_{-\infty}^\infty  e^{2\pi i  k|z_1 - z_2| \cos\theta } \frac{2i}{|k/2|^s} |k|dkd\theta,
\ea
so that
\ba
\langle \phi(z_1) \phi(z_2) \rangle_{b,\mathbb{C}} &=& 4i \int_0^\pi \int_{-\infty}^\infty  e^{2\pi i |k||z_1 - z_2| \cos\theta } \left|\frac{k}{2}\right|^{1-s}dkd\theta\\
&=&8i \int_0^\pi \int_0^\infty e^{2\pi |k||z_1 - z_2| \cos\theta } \left|\frac{k}{2}\right|^{1-s}dkd\theta \\
&=& \frac{4i\Gamma\lb |\cdot|^{2-s}\rb}{2^{1-s}|z_1-z_2|^{2-s}} \int_0^\pi \frac{d\theta}{ |\cos\theta |^{2-s}} \\
&=&4 i \frac{\sqrt{\pi}\Gamma_E\lb \frac{s}{2}-\frac{1}{2} \rb}{\Gamma_E\lb \frac{s}{2} \rb} \frac{\Gamma\lb |\cdot|^{2-s}\rb}{2^{1-s}} |z_1-z_2|^{s-2} \\
&=&4 i \frac{\Gamma\lb |\cdot|^{\frac{s}{2}- \frac{1}{2}} \rb}{\lb 1 + \tan\frac{\pi s}{4} \rb\Gamma\lb |\cdot|^\frac{s}{2} \rb} \frac{\Gamma\lb |\cdot|^{2-s}\rb}{2^{1-s}} |z_1-z_2|^{s-2}.
\ea
The Legendre duplication formula for the Gelfand-Graev Gamma function $\Gamma\lb |\cdot|^{2-s} \rb$ reads
\be
\Gamma(|\cdot|^{2-s}) = 2^{1-s} \Gamma \lb|\cdot|^{1-\frac{s}{2}}\rb \Gamma\lb |\cdot|^{\frac{3}{2}-\frac{s}{2}}\rb\lb 1+ \cot\frac{\pi s}{4}\rb,
\ee
and so we obtain
\ba
\langle \phi(z_1) \phi(z_2) \rangle_{b,\mathbb{C}} &=& 4  i \frac{ \Gamma^2 \lb|\cdot|^{1-\frac{s}{2}}\rb}{\lb 1 + \tan\frac{\pi s}{2} \rb} \lb 1+ \cot\frac{\pi s}{4}\rb |z_1-z_2|^{s-2}\\
&=&4 \Gamma \lb|\cdot|^{1-\frac{s}{2}}\rb \Gamma \lb|\cdot|^{1-\frac{s}{2}}\sgn_{-1}(\cdot)\rb |z_1-z_2|^{s-2}.
\ea
\end{proof}

\subsection{The main observation}
We have obtained:
\ba
\langle \phi(z_1)\phi(z_2) \rangle_b\lb\frac{s}{2}\rb &=& 2i \Gamma\lb |\cdot|^{1-\frac{s}{2}} \rb |z_2-z_1|^{\frac{s}{2}-1}, \\
\langle \phi(z_1)\phi(z_2) \rangle_f\lb\frac{s}{2}\rb &=& 2 \Gamma\lb |\cdot|^{1-\frac{s}{2}} \sgn_{-1}(\cdot) \rb |z_1-z_2|^{\frac{s}{2}-1} \sgn_{-1}\lb z_1 - z_2 \rb, \\
\langle \phi(z_1) \phi(z_2) \rangle_{b,\mathbb{C}} \lb s \rb &=& 4 \Gamma \lb|\cdot|^{1-\frac{s}{2}}\rb \Gamma \lb|\cdot|^{1-\frac{s}{2}}\sgn_{-1}(\cdot)\rb |z_1-z_2|^{s-2},
\ea
so that for $z_{1,2}$ restricted to the real axis we have
\be
\langle \phi(z_1) \phi(z_2) \rangle_{b,\mathbb{C}} \lb s \rb = -i \langle \phi(z_1)\phi(z_2) \rangle_b\lb\frac{s}{2}\rb \frac{ \langle \phi(z_1)\phi(z_2) \rangle_f\lb\frac{s}{2}\rb }{\sgn_{-1}\lb z_1 - z_2 \rb}.
\ee

\end{document}